\documentclass{article}

\usepackage{graphicx,setspace}
\usepackage[margin=1.25in]{geometry}
\usepackage{caption,subfig}
\usepackage{amsmath,amssymb,amsthm}
\usepackage{caption}
\usepackage{tabularx,multirow}
\usepackage{appendix}

\newcommand{\R}{{I \!\! R}}

\newtheorem{prop}{Proposition}
\newtheorem{defn}{Definition}

\begin{document}
\title{Some properties of a production function}
\author{C. Chilarescu}
\date{}
\maketitle
\centerline{\it University of Lille, France and West University of Timisoara, Romania}
\centerline{\it E-mail: Constantin.Chilarescu@univ-lille.fr}
\maketitle
\begin{abstract}
\noindent We examine the new production function developed by Chilarescu, and prove that under certain restrictions, the values of the elasticity can also be less than one. We will also prove that under certain restrictions on the parameters, the production function satisfies the Inada conditions.
\end{abstract}
\date{}
\maketitle
\section{Introduction}
A few years ago, the existence of a new production function with variable elasticity of substitution was proven by Chilarescu $(2021)$. The main aim of that paper was to develop a new production function, not only with variable elasticity of substitution, but also whose values are greater than one. In fact, the author tried to find a new production function that conforms to the theory developed by Piketty and Saez in an article published in  $(2014)$. Piketty and Saez analyzed the effect of the average annual real rate of return $(r)$ and the capital income ratio $\beta$, on the share of capital income in the national income, defined as $\alpha=r\beta$. In the standard economic model with perfectly competitive markets, $r$ is equal to the marginal product of capital. As the volume of capital $\beta$ rises, the marginal product $r$ tends to decline. If the standard hypothesis in economics are assumed, that is a unitary elasticity, then the fall in $r$, exactly offsets the rise in $\beta$, so that the capital share  is a technological constant. In recent decades, rich countries have experienced both a rise in $\beta$ and a rise in $\alpha$, which suggests that $\sigma$ is somewhat larger than one. This hypothesis has already been suggested by Sato and Hoffman $(1968)$.

But, looking more deeply at the new production function developed by Chilarescu, we observe that under certain restrictions, the values of the elasticity can also be less than one. We will analyze later in the present paper, in more details this new production function and we will clarify some essential of its properties. We will prove that under certain restrictions on the parameters, the production function satisfies the Inada conditions.

The first paper that tried to examine the implications of the Inada conditions on the asymptotic behaviour of the production functions was the paper of Barelli and Pess\^{o}a $(2003)$. Barelli and Pess\^{o}a proved that Inada conditions imply asymptotic Cobb-Douglas behavior of the production function. This result was was corrected few years later by Litina and Palivos $(2008)$ by putting that only the elasticity of substitution is equal to one. In two recent papers, de la Fonteijne $(2015)$ and Ozkaya $(2021)$ state that they have proven that the Inada conditions imply not only that the elasticity of substitution is bounded but also that the behavior of the production function considered is asymptotic Cobb-Douglas. More than that, Ozkaya proved that the parameters of the asymptotic Cobb-Douglas production function are endogenously determined by the original production function.

The paper is organized as follows. The first section is this introduction. In the next section we present our main result. We prove that under relatively reasonable restrictions the production function satisfies the Inada conditions. In the final section we present some asymptotical properties, numerical simulations and  conclusions.
\section{Inada conditions}
In this section we examine in more details the properties of the production function
\begin{equation}\label{eqcpf}
y=f(k)=A k^\theta\left[\alpha k^{\psi}+\beta\right]^{\omega},\; A > 0, \alpha > 0,\;\beta > 0.
\end{equation}
For this production function, the elasticity of substitution is given by the following equation
\begin{equation}\label{eqes}
\sigma(k)=\frac{\left[\alpha\left(\theta+\omega\psi\right)k^{\psi}
+\beta\theta\right]\left[\alpha\left(1-\theta-\omega\psi\right)k^{\psi}+\beta\left(1-\theta\right)\right]}
{\left[\alpha\left(\theta+\omega\psi\right)k^{\psi}
+\beta\theta\right]\left[\alpha\left(1-\theta-\omega\psi\right)k^{\psi}
+\beta\left(1-\theta\right)\right]-\alpha\beta\omega\psi^2k^\psi}.
\end{equation}
As we can observe, the fact that its values are superiors or inferiors to one depends also on the sign of the parameter $\omega$.
Now we can highlight the properties of this production function. To do this we need the following definition.
\begin{defn}
A production function $f : \R_+ \rightarrow \R_+$, expressed in intensive form, satisfies the Inada conditions if the following restrictions hold:
\begin{enumerate}
  \item $f(k) \geq 0,$ for all $k \geq 0$,\;$f(0) = 0$,\;\mbox{and}\;$\lim\limits_{k\rightarrow +\infty}f(k) = +\infty$,
  \item $f^{\prime}(k) > 0,$ for all $k > 0$,\;$\lim\limits_{k\rightarrow 0}f^{\prime}(k) = +\infty$,\;\mbox{and}\;$\lim\limits_{k\rightarrow +\infty}f^{\prime}(k) = 0$,
  \item $ f^{\prime\prime}(k) < 0$, for all $k > 0$.
\end{enumerate}
\end{defn}
Our main contribution is presented in the following proposition.
\begin{prop}
If the parameters of the production function defined by equation \eqref{eqcpf} lie in the following set:
\begin{equation}\label{setrestr}
\Phi \equiv \left\{\varphi\in\Phi | A >0, \alpha > 0, \beta > 0, \theta\in (0, 1),0 < \theta + \omega\psi <1, \psi < 1,0 < \omega\psi < 1\right\}
\end{equation}
where $\omega$ and $\psi$ have the same sign, then it satisfies the Inada conditions.
\end{prop}
\begin{proof}
The first and the second derivatives have the following formulas:
\begin{equation}\label{eqcfdpf}
f^{\prime}(k)=\frac{f(k)\left[\alpha\left(\theta+\omega\psi\right) k^\psi+\beta\theta\right]}{k\left(\alpha k^\psi+\beta\right)},
\end{equation}
\begin{equation}\label{eqcsdpf}
f^{\prime\prime}(k)=-\frac{f(k)g(k)}{{k}^{2}\left(\alpha{k}^{\psi}+\beta\right)^{2}},
\end{equation}
where
$$g(k)=\alpha^{2}\left(\theta+\omega\psi\right)\left(1-\theta-\omega
\psi\right)k^{2\psi}+\alpha\beta
 \left[\psi\omega\left(1-\psi\right)+2\theta\left(1-\theta-\omega\psi\right)\right]k^\psi+\beta^
2\theta\left(1-\theta\right).$$
Examining the equations \eqref{eqcpf}, \eqref{eqcfdpf} and \eqref{eqcsdpf}, we deduce that the parameters of the production function must respect the following restrictions:
\begin{enumerate}
  \item $A > 0$, $\alpha > 0$, $\beta > 0$;
  \item $1 > \theta+\omega\psi > 0$, $\beta\theta > 0$;
  \item $\omega\psi(1-\psi) > 0$, $\theta(1-\theta) > 0$.
\end{enumerate}
Consequently, from these restrictions we can claim that the parameters lie in the set \eqref{setrestr}.
From the definition of the set $\Phi$, we deduce that $\omega$ and $\psi$ will be both positive or both negative. It is just a simply exercise to prove that $\lim\limits_{k\rightarrow 0}f(k) = 0.$
Let us now examine the following limit.
$$\lim\limits_{k\rightarrow 0}f^{\prime}(k) =\lim\limits_{k\rightarrow 0}\frac{f(k)\left[\alpha\left(\theta+\omega\psi\right) k^\psi+\beta\theta\right]}{k\left(\alpha k^\psi+\beta\right)} = \lim\limits_{k\rightarrow 0}\frac{f(k)}{k}\lim\limits_{k\rightarrow 0}\frac{\alpha\left(\theta+\omega\psi\right) k^\psi+\beta\theta}{\alpha k^\psi+\beta}= I_1 I_2.$$
The second limit depends on of the value of the parameter $\psi$. Thus, if $\psi > 0$, then the value of the limit is equal to $\theta$ and if $\psi < 0$, then the value of the limit is equal to $\theta+\omega\psi$.
Let us now examine the first limit.
$$\lim\limits_{k\rightarrow 0}\frac{f(k)}{k}=\lim\limits_{k\rightarrow 0}\frac{A \left[\alpha k^{\psi}+\beta\right]^{\omega}}{k^{1-\theta}}.$$
As a consequence of the above results, the value of this limit is equal to $+\infty$ for any value of $\psi$ and therefore, we have $\lim\limits_{k\rightarrow 0}f^{\prime}(k) = +\infty$.
Finally, we need to evaluate the following limit:
$$\lim\limits_{k\rightarrow \infty}f^{\prime}(k)=\lim\limits_{k\rightarrow \infty}\frac{f(k)\left[\alpha\left(\theta+\omega\psi\right) k^\psi+\beta\theta\right]}{k\left(\alpha k^\psi+\beta\right)}=\lim\limits_{k\rightarrow \infty}\frac{A \left[\alpha k^{\psi}+\beta\right]^{\omega}}{k^{1-\theta}}\lim\limits_{k\rightarrow \infty}\frac{\alpha\left(\theta+\omega\psi\right) k^\psi+\beta\theta}{\alpha k^\psi+\beta}=I_1 I_2$$
As in the previous case, the second limit $I_2$ has a finite limit, equal to $\theta+\omega\psi$ if $\psi > 0$ or equal to $\theta$, if $\psi < 0$. For the first limit we have
$$I_1 = \lim\limits_{k\rightarrow \infty}\frac{A \left[\alpha k^{\psi}+\beta\right]^{\omega}}{k^{1-\theta}}=
\left\{\begin{array}{c}
         \lim\limits_{k\rightarrow \infty}\frac{A \left[\alpha +\beta k^{-\psi}\right]^{\omega}}{k^{1-\theta -\omega\psi}} = 0, \;\mbox{if}\; \psi > 0,\\\\
        \lim\limits_{k\rightarrow \infty}\frac{A \left[\alpha k^{\psi}+\beta\right]^{\omega}}{k^{1-\theta}}, \; = 0,\;\mbox{if}\; \psi < 0,
       \end{array}\right.$$
and therefore we have $\lim\limits_{k\rightarrow \infty}f^{\prime}(k) = 0$,
and thus we have proved that this production function satisfies the Inada conditions.
\end{proof}
\section{Asymptotic properties and numerical simulations}
In this section we will examine some asymptotic properties of this production function. As was proved by Ozkaya $(2021)$, if a production function expressed in intensive form satisfies the Inada conditions, then it converges asymptotically to a Cobb-Douglas production function.
Applying now the result obtained by Ozkaya $(2021)$ in Proposition $3$ we get:
$$\alpha_z = \lim\limits_{k\rightarrow 0}\frac{kf^{\prime}(k)}{f(k)}=\lim\limits_{k\rightarrow 0}\frac{kf(k)\left[\alpha\left(\theta+\omega\psi\right) k^\psi+\beta\theta\right]}{kf(k)\left(\alpha k^\psi+\beta\right)}=\lim\limits_{k\rightarrow 0}\frac{\alpha\left(\theta+\omega\psi\right) k^\psi+\beta\theta}{\alpha k^\psi+\beta}.$$
$$\beta_z = \lim\limits_{k\rightarrow +\infty}\frac{kf^{\prime}(k)}{f(k)}=\lim\limits_{k\rightarrow +\infty}\frac{kf(k)\left[\alpha\left(\theta+\omega\psi\right) k^\psi+\beta\theta\right]}{kf(k)\left(\alpha k^\psi+\beta\right)}=\lim\limits_{k\rightarrow +\infty}\frac{\alpha\left(\theta+\omega\psi\right) k^\psi+\beta\theta}{\alpha k^\psi+\beta}.$$
As proven above, these limits depend on the sign of $\psi$ and thus we obtain
\begin{equation}\label{eqabz}
\alpha_z = \left\{\begin{array}{c}
         \theta, \;\mbox{if}\; \psi > 0,\\\\
        \theta+\omega\psi,\;\mbox{if}\; \psi < 0.
       \end{array}\right.,\;\beta_z = \left\{\begin{array}{c}
         \theta+\omega\psi, \;\mbox{if}\; \psi > 0,\\\\
        \theta,\;\mbox{if}\; \psi < 0.
       \end{array}\right.
\end{equation}
\begin{equation}\label{eqABz0}
A_z = \left\{\begin{array}{c}
\lim\limits_{k\rightarrow 0}\frac{A k^\theta\left[\alpha k^{\psi}+\beta\right]^{\omega}}{k^\theta} = A\lim\limits_{k\rightarrow 0}\left[\alpha k^{\psi}+\beta\right]^{\omega} = A\beta^\omega,
 \;\mbox{if}\; \psi > 0,\\\\
 \lim\limits_{k\rightarrow 0}\frac{A k^\theta\left[\alpha k^{\psi}+\beta\right]^{\omega}}{k^{\theta+\omega\psi}} = A\lim\limits_{k\rightarrow 0}\frac{\left[\alpha k^{\psi}+\beta\right]^{\omega}}{k^{\omega\psi}} = A\alpha^\omega,
  \;\mbox{if}\; \psi < 0.
   \end{array}\right.
\end{equation}
\begin{equation}\label{eqABzin}
B_z = \left\{\begin{array}{c}
\lim\limits_{k\rightarrow +\infty}\frac{A k^\theta\left[\alpha k^{\psi}+\beta\right]^{\omega}}{k^{\theta+\omega\psi}} = A\lim\limits_{k\rightarrow +\infty}\frac{\left[\alpha k^{\psi}+\beta\right]^{\omega}}{k^{\omega\psi}} = A\alpha^\omega,
 \;\mbox{if}\; \psi > 0,\\\\
 \lim\limits_{k\rightarrow +\infty}\frac{A k^\theta\left[\alpha k^{\psi}+\beta\right]^{\omega}}{k^{\theta}} = A\lim\limits_{k\rightarrow +\infty}\left[\alpha k^{\psi}+\beta\right]^{\omega} = A\beta^\omega,
  \;\mbox{if}\; \psi < 0.
   \end{array}\right.
\end{equation}
Thus, the limit production function take the following form
\begin{equation}\label{eqpf0}
g\left(k,\;k \rightarrow 0\right) \left\{\begin{array}{c}g_1= A\beta^\omega k^\theta,
 \;\mbox{if}\; \psi > 0,\\\\
  g_2=A\alpha^\omega k^{\theta+\omega\psi},
  \;\mbox{if}\; \psi < 0,
   \end{array}\right.\;f\left(k,\;k \rightarrow +\infty\right) \left\{\begin{array}{c} f_1 =A\alpha^\omega k^{\theta+\omega\psi},
 \;\mbox{if}\; \psi > 0,\\\\
  f_2 =A\beta^\omega k^{\theta},
  \;\mbox{if}\; \psi < 0.
   \end{array}\right.
\end{equation}
The following numerical simulations are the consequence of the two possible cases generated by the two alternative values of the parameter $\omega$. The benchmark values for economy we consider are the following:

[$Case\; 1.$\;$\omega > 0$]:\; $\theta=0.8$, $\psi=0.9$, \;\;\;$\omega=0.2$, \;\;\;$\alpha=0.2$, $\beta=0.8$, $A=1.05$.

\vspace{0.2cm}

[$Case\; 2.$\;$\omega < 0$]:\;  $\theta=0.8$, $\psi=-0.9$, $\omega=-0.2$, $\alpha=0.2$, $\beta=0.8$, $A=1.05$.

The evolution of the elasticity of substitution and the transitional dynamics for the two production functions $f$ and $f_1$, of the first case are presented in the following three graphics.
\begin{center}
\includegraphics[width=6.5cm]{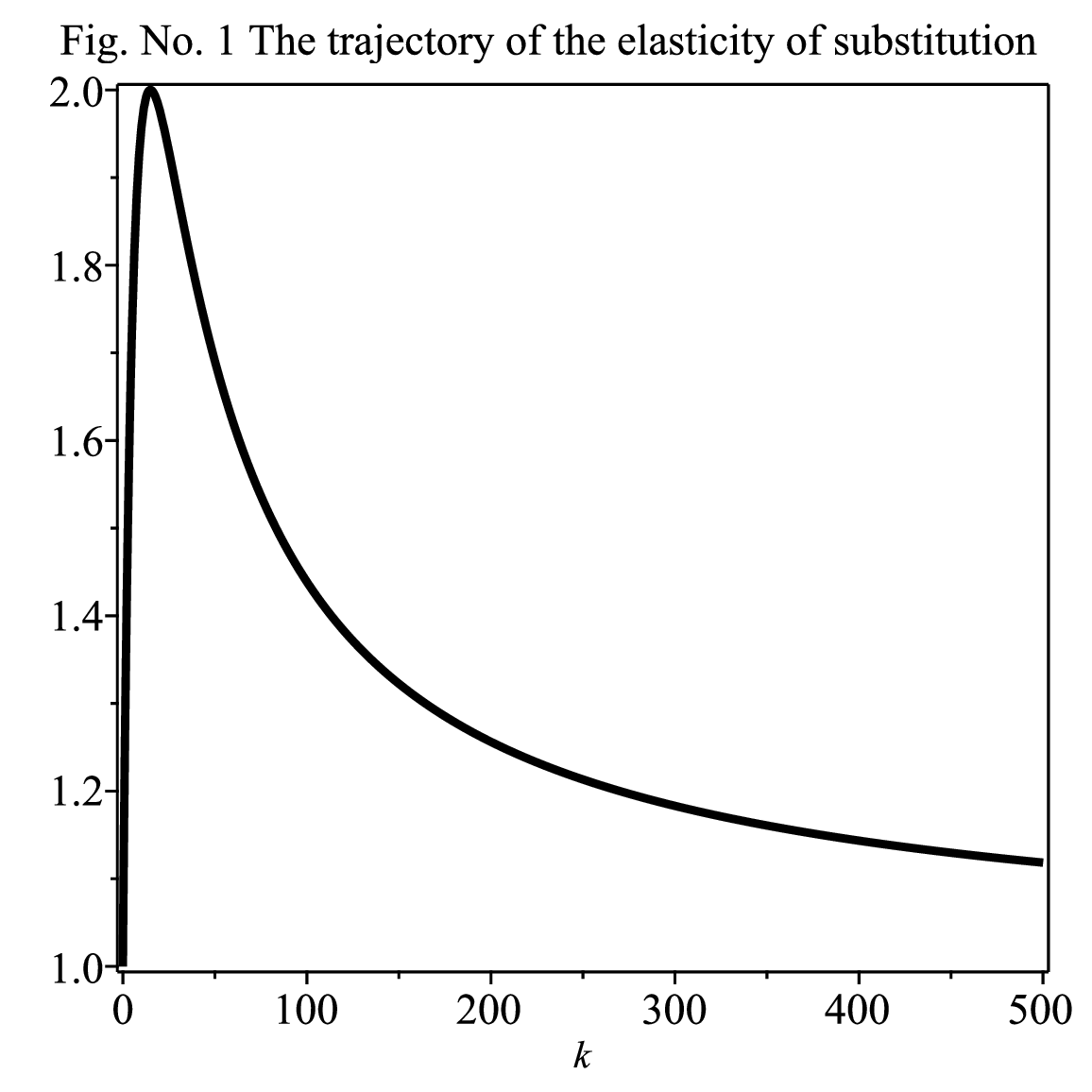}
\end{center}

\begin{center}
\includegraphics[width=6.5cm]{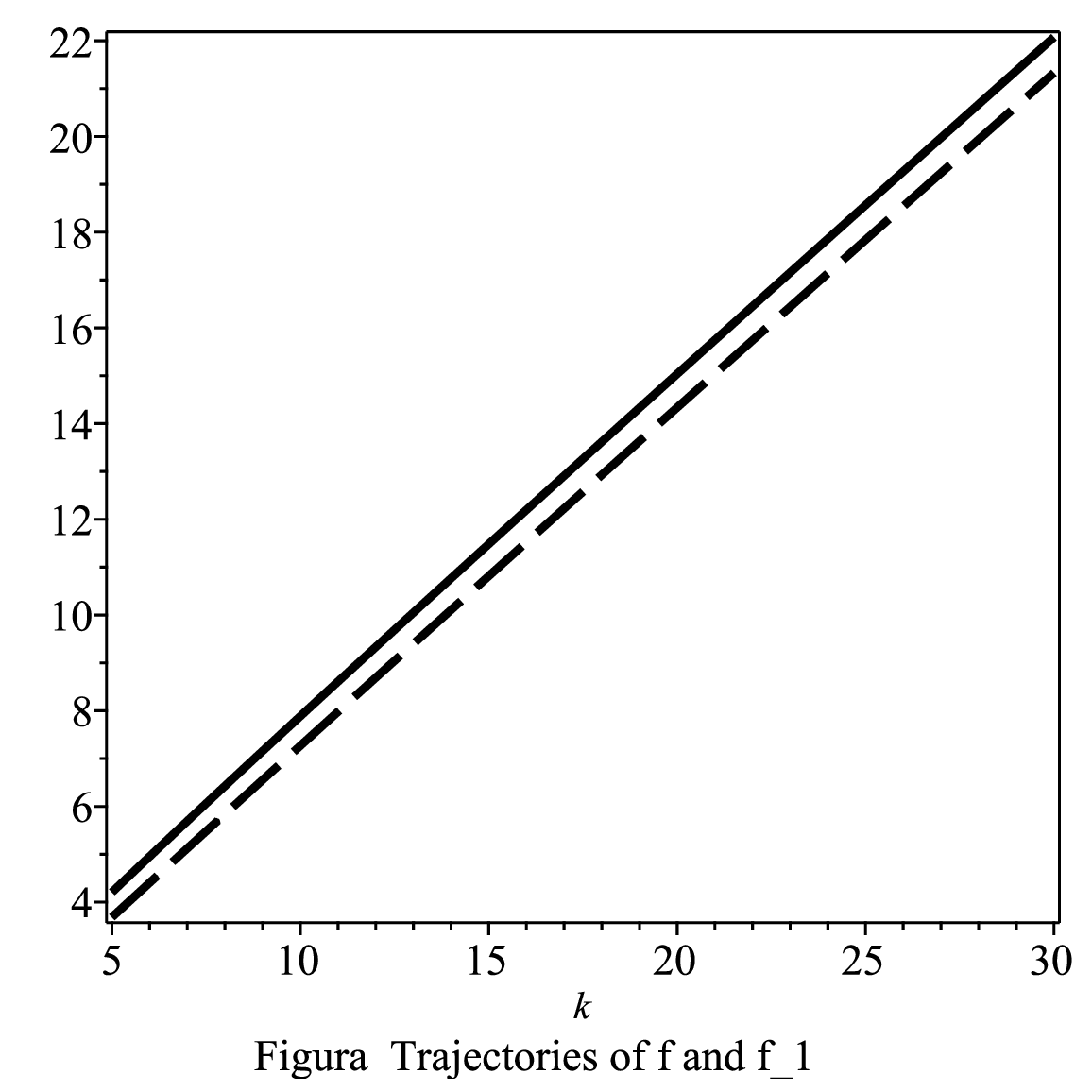}\;\;\;\includegraphics[width=6.5cm]{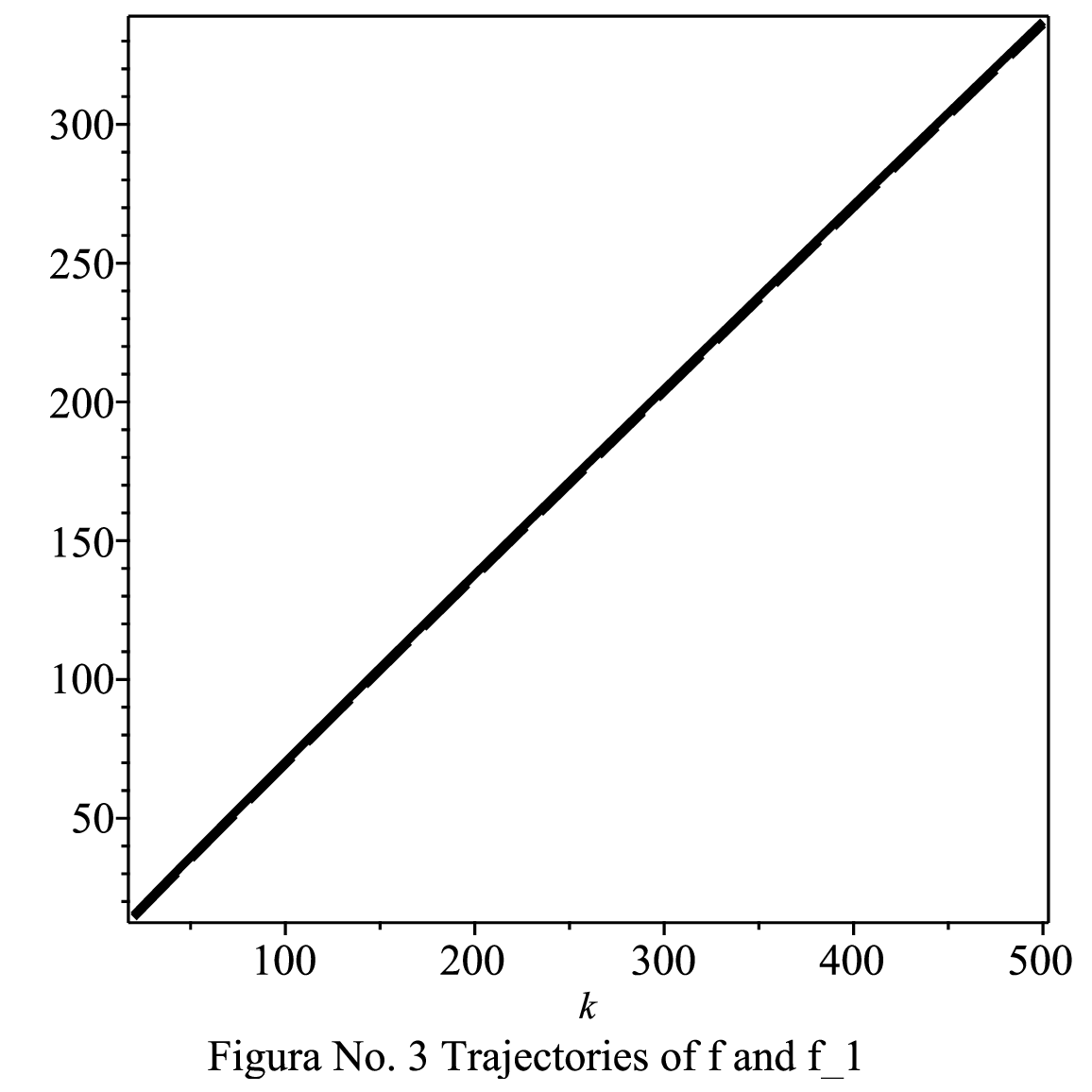}
\end{center}
As we can observe, there are significant differences between the two trajectories over the period whose elasticity of substitution is significantly greater than one, and the two trajectories almost coincides when $k$ tends to infinity.

The evolution of the elasticity of substitution and the transitional dynamics for the two production functions $f$ and $f_2$, of the second case are presented in the following three graphics.
\begin{center}
\includegraphics[width=6.5cm]{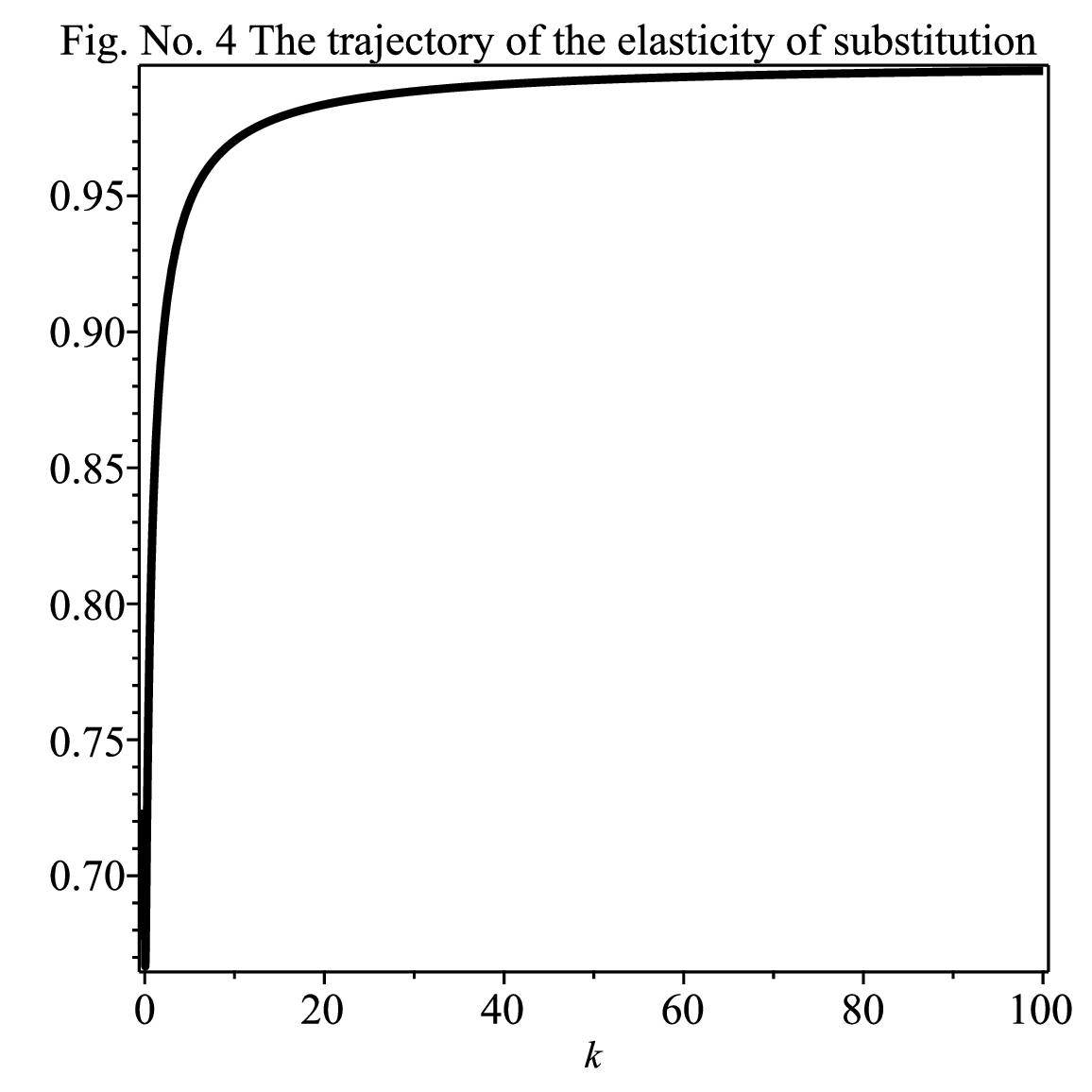}
\end{center}

\begin{center}
\includegraphics[width=6.5cm]{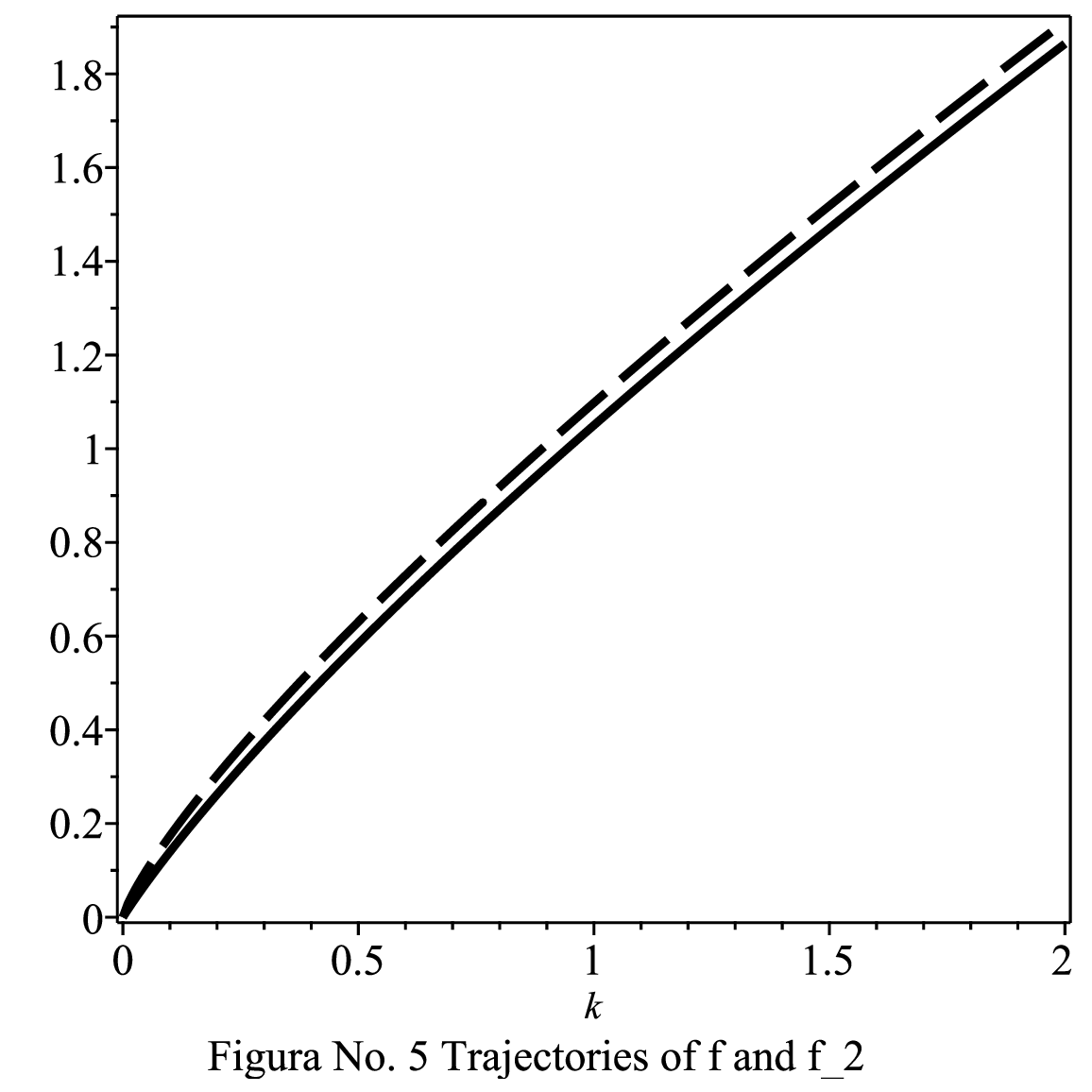}\;\;\;\includegraphics[width=6.5cm]{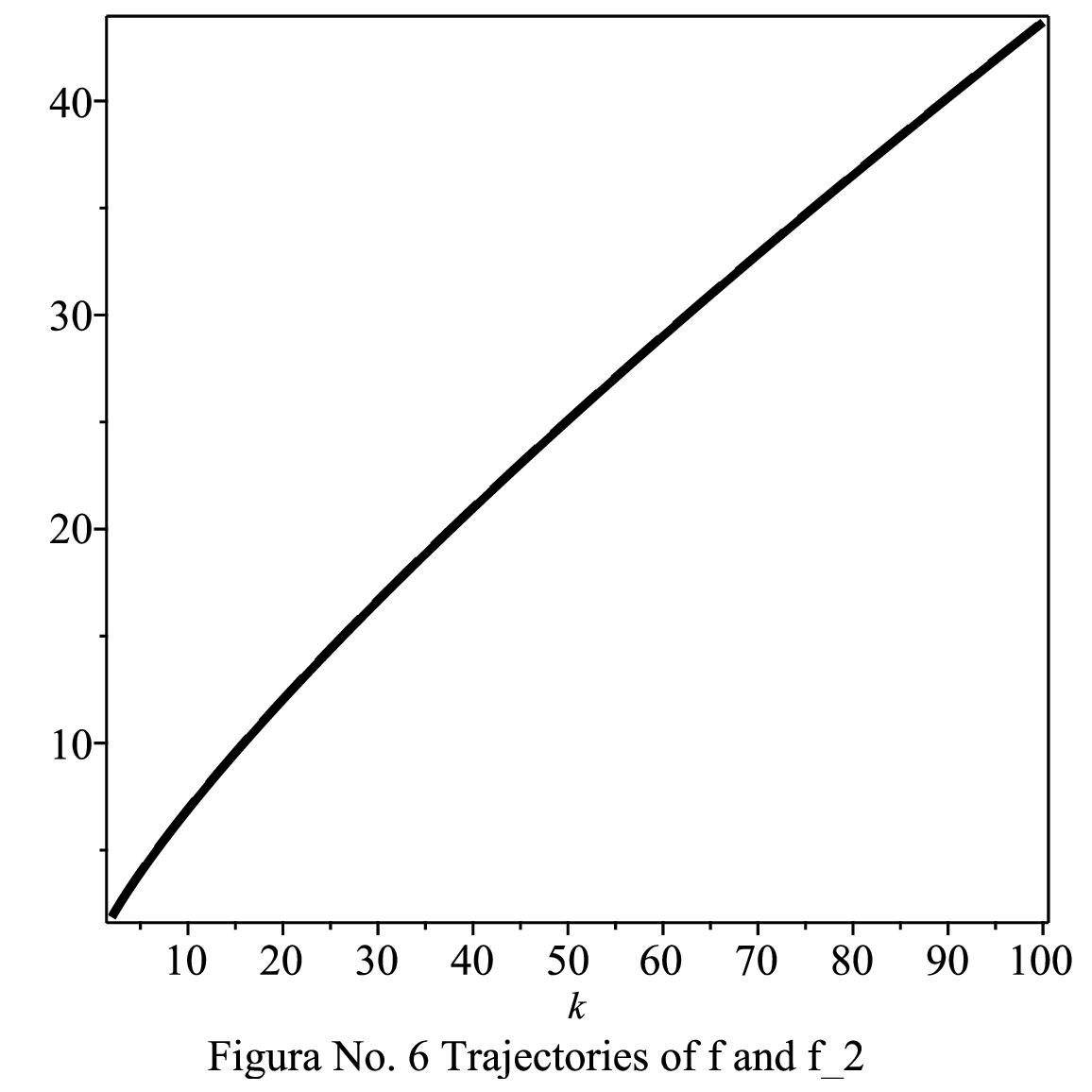}
\end{center}
As we can observe, there are significant differences between the two trajectories over the period whose elasticity of substitution is significantly less than one, and the two trajectories almost coincides when $k$ tends to infinity.

\end{document}